%% file: altruists.tex
\documentclass[11pt]{article}

\usepackage{amsmath}
\usepackage{amssymb}
\usepackage{subfigure}
\usepackage{latexsym} 			
\usepackage{amsfonts}
\usepackage{fullpage}
\usepackage{journals}
\usepackage{cite,caption}
\usepackage{algo}
\usepackage{amsthm}
\usepackage{graphicx}
\usepackage{color}

\usepackage{hyphenat}
\usepackage{verbatim}

\usepackage{algorithm}
\usepackage{algorithmic}
\usepackage{xcolor}
\usepackage{xspace}

\newtheorem{prop}{Proposition}

\newtheorem{obs}[prop]{Observation}
\newtheorem{exm}[prop]{Example}

\newtheorem{cor}[prop]{Corollary}
\newtheorem{thm}[prop]{Theorem}

\newcommand{\stratSpace}{\mathcal{S}}

\newcommand{\DeltaPhi}{\ensuremath{\sum_{e \in S_i \setminus S_i'} a_e n_e  - \sum_{e \in S_i' \setminus S_i} a_e (n_e +1)}}
\newcommand{\DeltaC}{\ensuremath{ \sum_{e \in S_i \setminus S_i'} (2a_e n_e - a_e) - \sum_{e \in  S_i' \setminus S_i} (2 a_e n_e + a_e )}}
\newcommand{\NP}{{\sf NP}}
\newcommand{\PLS}{{\sf PLS}}

\newcommand{\junk}[1]{}

\title{Altruism in Atomic Congestion Games}%
\author{Martin Hoefer\thanks{Computer Science Department, Stanford
    University, USA. Supported by a fellowship within the
    Postdoc-Program of the German Academic Exchange Service (DAAD).}
  \and Alexander Skopalik\thanks{Dept. of Computer Science, RWTH
    Aachen University, Germany. Supported in part by the German
    Israeli Foundation (GIF) under contract 877/05.}}

\begin{document}
\maketitle
\thispagestyle{empty}
\begin{abstract}
  This paper studies the effects of introducing altruistic agents into
  atomic congestion games. Altruistic behavior is modeled by a
  trade-off between selfish and social objectives. In particular, we
  assume agents optimize a linear combination of personal delay of a
  strategy and the resulting social cost. Our model can be embedded in
  the framework of congestion games with player-specific latency
  functions. Stable states are the Nash equilibria of these games, and
  we examine their existence and the convergence of sequential
  best-response dynamics. Previous work shows that for symmetric
  singleton games with convex delays Nash equilibria are guaranteed to
  exist. For concave delay functions we observe that there are games
  without Nash equilibria and provide a polynomial time algorithm to
  decide existence for symmetric singleton games with arbitrary delay
  functions. Our algorithm can be extended to compute best and worst
  Nash equilibria if they exist. For more general congestion games
  existence becomes \NP-hard to decide, even for symmetric network
  games with quadratic delay functions. Perhaps surprisingly, if all
  delay functions are linear, then there is always a Nash equilibrium
  in any congestion game with altruists and any better-response
  dynamics converges.

  In addition to these results for uncoordinated dynamics, we consider
  a scenario in which a central altruistic institution can motivate
  agents to act altruistically. We provide constructive and hardness
  results for finding the minimum number of altruists to stabilize an
  optimal congestion profile and more general mechanisms to
  incentivize agents to adopt favorable behavior.
\end{abstract}


\section{Introduction}
Algorithmic game theory has been focused on game-theoretic models for
a variety of important applications in the Internet, e.g.~selfish
routing~\cite{RoughTar02,Awerbuch05,Christo05}, network
creation~\cite{Anshe04}, as well as aspects of
e-commerce~\cite{Guruswami05} and social networks~\cite{Ghosh08}. A
fundamental assumption in these games, however, is that all agents are
\emph{selfish}. Their goals are restricted to optimizing their direct
personal benefit, e.g.~their personal delay in a routing game. The
assumption of selfishness in the preferences of agents is found in the
vast majority of present work on economic aspects of the
Internet. However, this assumption has been repeatedly questioned by
economists and psychologists. In experiments it has been observed that
participant behavior can be quite complex and contradictive to
selfishness~\cite{Ledyard97,Levine98}. Various explanations have been
given for this phenomenon, e.g.~senses of fairness~\cite{Fehr99},
reciprocity among agents~\cite{Gintis05}, or spite and
altruism~\cite{Levine98,Eshel98}.

Prominent developments in the Internet like Wikipedia, open source
software development, or Web 2.0 applications involve or explicitly
rely on voluntary participation and contributions towards a joint
project without direct personal benefit. These examples display forms
of \emph{altruism}, in which agents accept certain personal burdens
(e.g.~by investing time, attention, and money) to improve a common
outcome. While malicious behavior has been considered recently for
instance in nonatomic routing~\cite{Karakostas07,Babaioff07,Chen08},
virus incoulation~\cite{Moscibroda06Evil}, or bayesian congestion
games~\cite{Gairing08}, a deeper analysis of the effects of altruistic
agents on competitive dynamics in algorithmic game theory is still
missing.

We consider and analyze a model of altruism inspired by
Ledyard~\cite[p. 154]{Ledyard97}, and recently studied for non-atomic
routing games by Chen and Kempe~\cite{Chen08}. Each agent $i$ is
assumed to be partly selfish and partly altruistic. Her incentive is
to optimize a linear combination of personal cost and social cost,
given by the sum of cost values of all agents. The strength of
altruism of each agent $i$ is captured by her \emph{altruism level}
$\beta_i \in [0,1]$, where $\beta_i = 0$ results in a purely selfish
and $\beta_i = 1$ in a purely altruistic agent.

Chen and Kempe~\cite{Chen08} proved that in non-atomic routing games
Nash equilibria are always guaranteed to exist, even for partially
spiteful users, and analyzed the price of anarchy for parallel link
networks. In our paper, we conduct the first study of altruistic
agents in atomic congestion games, a well-studied model for resource
sharing. A standard congestion game is given by a set $N$ of myopic
selfish users and a set $E$ of resources. Each resource $e$ has a
non-decreasing delay function $d_e$. Every agent $i$ can pick a
strategy $S_i$ from a set of possible strategies $\stratSpace_i
\subseteq 2^E$, which means she allocates the set $S_i$ of resources
(e.g.~a path in a network). She then experiences a delay corresponding
to the total delay on all resources in $S_i$, which in turn depends on
the number of agents that allocate each resource. Each agent strives
to pick a strategy minimizing her experienced delay. A stable state in
such a game is a pure Nash equilibrium, in which each agent picks
exactly one strategy, and no agent can decrease her delay by
unilaterally changing her strategy. The study of congestion games
received a lot of attention in recent years, mostly because of the
intuitive formulation and their appealing analytical properties. In
particular, they always possess a pure Nash equilibrium and every
sequential better-response dynamics converges.

As one might expect, the presence of altruists can significantly alter
the convergence and existence guarantees of pure Nash equilibria in
congestion games. After a formal definition of \emph{congestion games
  with altruists} in Section~\ref{sect:model}, we concentrate on pure
equilibria and leave a study of mixed Nash equilibria for future
work. Our results are as follows.

It is a simple exercise to observe that even in a \emph{singleton
  game}, in which each strategy consists of a single resource, and for
\emph{symmetric} agents, where each agent has the same set of
strategies, a Nash equilibrium can be absent. This is the case even
for \emph{pure altruists and egoists}, i.e.~a population of agents
which are either purely altruistic or purely selfish and their
$\beta_i \in \{0,1\}$. However, we show in
Section~\ref{sect:singleton} that such games admit a polynomial time
algorithm to decide the existence problem. Furthermore, our algorithm
can be adapted to compute the Nash equilibrium with best and worst
social cost if it exists, for any agent population with a constant
number of different altruism levels.

For slightly more general \emph{asymmetric singleton games}, in which
strategy spaces of agents differ, we show in
Section~\ref{sect:singleton} that deciding the existence of Nash
equilibria becomes \NP-hard. Nevertheless, for the important subclass
of convex delay functions, previous results imply that for any agent
population a Nash equilibrium exists and can be obtained in polynomial
time. In contrast, we show in Section~\ref{sect:general} that
convexity of delay functions is not sufficient for more general
games. In particular, even for \emph{symmetric network games}, in
which strategies represent paths through a network, quadratic delay
functions and pure altruists, Nash equilibria can be absent and
deciding their existence is \NP-hard. Perhaps surprisingly, if all
delay functions are \emph{linear}, then there is a potential
function. Thus, for every agent population Nash equilibria exist and
better-response dynamics converges.

In addition to these results for uncoordinated dynamics, in
Section~\ref{sect:stabilize} we consider a slightly more coordinated
scenario, in which there is a central institution striving to obtain a
good outcome. An obvious way to induce favorable behavior is to
convince agents to act altruistically. In this context a natural
question is how many altruists are required to stabilize a social
optimum. This has been considered under the name ``price of optimum''
in~\cite{Kaporis06} for Stackelberg routing in nonatomic congestion
games. As a Nash equilibrium in atomic games is not necessarily
unique, we obtain two measures - an \emph{optimal stability
  threshold}, which is the minimum number of altruists such that there
is \emph{any} optimal Nash equilibrium, and an \emph{optimal anarchy
  threshold}, which asks for the minimum number of altruists such that
\emph{every} Nash equilibrium is optimal. For symmetric singleton
games, we adapt our algorithm for computing Nash equilibria to
determine both thresholds in polynomial time.

In our model the optimal anarchy threshold might not be well-defined
even for singleton games. If all agents are altruists, there are
suboptimal local optima in symmetric games with concave delays, or in
asymmetric games with linear delays. Hence, even by making all agents
altruists, the worst Nash equilibrium sometimes remains suboptimal. In
contrast, we adapt the idea of the optimal stability threshold to a
very general scenario, in which we can find a stable state with a
given, not necessarily optimal, congestion profile. Each agent has a
personalized \emph{stability cost} for accepting a strategy under the
given congestions. We provide an incentive compatible mechanism to
determine an allocation of agents to strategies with minimum total
stability cost. Unfortunately, such a general result is restricted to
the case of singleton games. Even for symmetric network games on
series-parallel graphs, we show that the problem of determining the
optimal stability threshold is \NP-hard.

\section{Model and Initial Results}
\label{sect:model}
We consider congestion games with altruists. A \emph{congestion game
  with altruists} $G$ is given by a set $N$ of $n$ agents and a set
$E$ of $m$ resources. Each agent $i$ has a set $\stratSpace_i
\subseteq 2^E$ of strategies. In a \emph{singleton} congestion game
each agent has only singleton strategies $\stratSpace_i \subseteq
E$. A vector of strategies $S = (S_1,\ldots,S_n)$ is called a
\emph{state}. For a state we denote by $n_e$ the congestion, i.e.~the
number of agents using a resource $e$ in their strategy. Each resource
$e$ has a \emph{latency} or \emph{delay} function $d_e(n_e)$, and the
\emph{delay for an agent} $i$ playing $S_i$ in state $S$ is $d_i(S) =
\sum_{e \in S_i} d_e(n_e)$. The \emph{social cost} of a state is the
total delay of all agents $c(S) = \sum_{i \in N} \sum_{e \in S_i}
d_e(n_e) = \sum_{e \in E} n_e d_e(n_e)$. Each agent $i$ has an
\emph{altruism level} of $\beta_i \in [0,1]$, and her \emph{individual
  cost} is $c_i(S) = \beta_i c(S) + (1-\beta_i) d_i(S)$. We call an
agent $i$ an \emph{egoist} if $\beta_i = 0$ and a
\emph{$\beta_i$-altruist} otherwise. A \emph{(pure) altruist} has
$\beta_i = 1$, a \emph{(pure) egoist} has $\beta_i = 0$. A game $G$
\emph{with only pure altruists and egoists} is a game, in which
$\beta_i \in \{0,1\}$ for all $i \in N$. A game $G$ is said to have
\emph{$\beta$-uniform} altruists if $\beta_i = \beta \in [0,1]$ for
every agent $i \in N$. A (pure) \emph{Nash equilibrium} is a state
$S$, in which no agent $i$ can unilaterally decrease her individual
cost by unilaterally changing her strategy. We exclusively consider
pure equilibria in this paper.

If all agents are egoists, the game is a regular congestion game,
which has an exact potential function $\Phi(S) = \sum_{e \in E}
\sum_{x=1}^{n_e} d_e(x)$~\cite{Rosenthal73}. Thus, existence of Nash
equilibria and convergence of iterative better-response dynamics are
guaranteed. Obviously, if all agents are altruists, Nash equilibria
correspond to local optima of the social cost function $c$ with
respect to a local neighborhood consisting of single player strategy
changes. Hence, existence and convergence are also guaranteed. This
directly implies the same properties for $\beta$-uniform games, in
which an exact potential function is $\Phi_\beta(S) = (1-\beta)
\Phi(S) + \beta c(S)$.

In general, however, Nash equilibria might not exist. 

\begin{obs}
  \label{obs:noNE}
  There are symmetric singleton congestion games with only pure
  altruists and egoists without a Nash equilibrium.
\end{obs}

\begin{exm} \rm
  \label{exm:noNE}
  Consider a game with two resources $e$ and $f$, three egoists and
  one (pure) altruist. The delay functions are $d_e(x) = d_f(x)$ with
  $d_e(1) = 4$, $d_e(2) = 8$, $d_e(3) = 9$, and $d_e(4) = 11$. Then,
  in equilibrium each resource must be allocated by at least one
  egoist. In case there are two agents on each resource, the social
  cost is 32. In this case the altruist is motivated to change as the
  resulting cost is 31. In that case, however, one of the egoists on
  the resource with congestion 3 has an incentive to change. Thus, no
  Nash equilibrium will evolve.
\end{exm}

Our interest is thus to characterize the games that have Nash
equilibria. Towards this end we observe that an altruistic congestion
game can be cast as a congestion game with player-specific latency
functions~\cite{Milchtaich96}. For simplicity consider a game with
only pure altruists and egoists. An altruist moves from $S_i$ to
$S'_i$ if the decrease in total delay $n_e d_e(n_e)$ on the resources
$e \in S_i - S'_i$ she is leaving exceeds the increase on resources $e
\in S'_i - S_i$ she is migrating to. Hence, altruists can be seen as
myopic selfish agents with $c_i(S) = d'_i(S) = \sum_{e \in S_i}
d'_e(n_e)$ with $d'_e(n_e) = n_e d_e(n_e) - (n_e - 1)d_e(n_e - 1)$,
for $n_e > 0$. We set $d'_e(0) = 0$. Naturally, a $\beta_i$-altruist
corresponds to a selfish agent with player-specific function $c_i(S) =
(1-\beta_i) d_i(S) + \beta_i d'_i(S)$. Thus, our games can be embedded
into the class of player-specific congestion games. For some classes
of these games it is known that Nash equilibria always exist. In
particular, non-existence in Example~\ref{exm:noNE} is due to the fact
that the individual delay function for the altruist is not
monotone. Monotonicity holds, in particular, if delay functions are
convex. In this case, it is known that for matroid games, in which the
strategy space of each agent is a matroid, existence of a Nash
equilibrium is guaranteed~\cite{Acker07}.

\begin{cor}{\cite{Milchtaich96, Acker07}}
  For any matroid congestion game with altruists and convex delay
  functions a Nash equilibrium exists and can be computed in
  polynomial time.
\end{cor}

\section{Singleton Congestion Games}
\label{sect:singleton}
In the previous section we have seen that there are symmetric
singleton congestion games with only pure altruists and egoists with
and without Nash equilibria. For this class of games we can decide the
existence of Nash equilibria in polynomial time. In addition, we can
compute a Nash equilibrium with minimum and maximum social cost if
they exist.

\begin{thm}
  \label{thm:singleALLES}
  For symmetric singleton games with only pure altruists and egoists
  there is a polynomial time algorithm to decide if a Nash equilibrium
  exists and to compute the best and the worst Nash equilibrium.
\end{thm}

\begin{proof}
  We first tackle the existence problem and present an approach
  similar to~\cite{Ieong05} based on dynamic programming. Suppose we
  are given a game $G$ with the set $N_0$ of $n_0$ egoists and the set
  $N_1$ of $n_1 = n-n_0$ altruists. For a state $S$ consider the set
  of resources $E_0 = \bigcup_{i \in N_0} S_i$ on which at least one
  egoist is located. The maximum delay of any resource on which an
  egoist is located is denoted $d_0^{max} = \max_{e \in E_0}
  d_e(n_e)$ and minimum delay of any resource if an additional agent
  is added $d_0^{min+} = \min_{e \in E} d_e(n_e + 1)$. Similarly,
  consider the set of resources $E_1 = \bigcup_{i \in N_1} S_i$. The
  maximum altruistic delay of any resource, on which an altruist is
  located, is denoted $d_1^{max} = \max_{e \in E_1} d'_e(n_e)$ and the
  minimum altruistic delay of any resource $d_1^{min+} = \min_{e \in
    E} d'_e(n_e + 1)$. A state is a Nash equilibrium if and only if
  \begin{equation}
    \label{eqn:NE}
    d_0^{max} \le d_0^{min+} \hspace{0.5cm} \mbox{ and } 
    \hspace{0.5cm} d_1^{max} \le d_1^{min+}. 
  \end{equation}
  This condition yields a separation property. Consider a Nash
  equilibrium, in which $n_{E',0}$ egoists and $n_{E',1}$ altruists
  are located on a subset $E' \subset E$ of resources. The Nash
  equilibrium respects the inequalities above for certain values
  $d_{0/1}^{max}$ and $d_{0/1}^{min+}$. Note that it is possible to
  completely change the assignment of agents in $E'$. If the new
  assignment respects the inequalities for the same values, it can be
  combined with the assignment on $E-E'$ and again a Nash equilibrium
  evolves.

  This property suggests the following approach to search for an
  equilibrium. Suppose the values for $d_{0/1}^{max}$ and
  $d_{0/1}^{min+}$ are given. Our algorithm adds resources $e$ one by
  one and tests the possible numbers of egoists and altruists that can
  be assigned to $e$. Suppose we have processed the resources from a
  subset $E'$ and have found the numbers of altruists and egoists, for
  which there is an assignment to resources $E'$ such that there is no
  violation of equations (\ref{eqn:NE}) for the given delay values. In
  this case, we know the feasible numbers of altruists and egoists
  that are left to be assigned to the remaining resources. Suppose we
  have marked these combinations of remaining agents in a boolean
  matrix $R$ of size $(n_0+1) \times (n_1+1)$. Here $r_{ij} = 1$ if
  and only if there is a feasible assignment of $n_0 - i$ egoists and
  $n_1 - j$ altruists to $E'$. For the new resource $e$ we now test
  all combinations $(n_{e,0}, n_{e,1})$ of altruists and egoists that
  can be allocated to $e$ such that the equations (\ref{eqn:NE})
  remain fulfilled. We then compile a new matrix $R'$ of the feasible
  combinations of remaining agents for the remaining resources $E - E'
  - \{e\}$. In particular, for each tuple $(n_{e,0}, n_{e,1})$ and
  each positive entry $r_{ij}$ of $R$ we check if $i - n_{e,0} \ge 0$
  and $j - n_{e,1} \ge 0$. If this holds, we set the entry of $R'$
  with index $(i - n_{e,0},j - n_{e,1})$ to 1. If $e$ is the last
  resource to be processed, we check if the resulting matrix $R'$ has
  a positive entry $r'_{0,0} = 1$. In this case a Nash equilibrium
  exists for the given values of $d_{\{0,1\}}^{max}$ and
  $d_{\{0,1\}}^{min+}$, otherwise it does not exist. Due to the
  separation property mentioned above, this approach succeeds to
  implicitly test all allocations that fulfill equations
  (\ref{eqn:NE}) for the given values. Finally, note that there are
  only at most $O(n_0^2n_1^2m^4)$ possible values for which we must
  run the algorithm.

  The separation property mentioned above also applies to the best or
  worst Nash equilibrium. In particular, consider the best Nash
  equilibrium $S$ that respects (\ref{eqn:NE}) for some fixed values
  $d_{\{0,1\}}^{max}$ and $d_{\{0,1\}}^{min+}$. Consider any subset of
  resources $E'$ with a number $n_{E',0}$ and $n_{E',1}$ of egoists
  and altruists, respectively. $S$ is the cheapest Nash equilibrium
  that respects (\ref{eqn:NE}) for the given values if and only if the
  assignment of $S$ in $E'$ is the cheapest assignment with $n_{E',0}$
  egoists and $n_{E',1}$ altruists that respects (\ref{eqn:NE}) for
  the values. Thus, we can adjust our approach as follows. For a set
  $E'$ of processed resources, instead of simply noting in $r_{ij}$
  that there is a feasible assignment to $E'$ that leaves $i$ egoists
  and $j$ altruists, we can remember the social cost of the cheapest
  of such assignments. Thus, the matrix $R$ is then a matrix of
  positive entries, for which we use a prohibitively large cost to
  identify infeasible combinations. When we compile a new matrix $R'$
  after testing all feasible assignments to a new resource $e$, we can
  denote in each entry the minimum cost that can be obtained for the
  respective combination. A similar argument works for computing the
  worst Nash equilibrium. This decides the existence question and
  finds the cost values of best and worst Nash equilibria. By tracing
  back the steps of the algorithm we can also discover the strategy
  choices of agents.
\end{proof}

Note that the previous proof can be extended to a constant number $k$
of different altruism levels. In this more general scenario we choose
the delay parameters for each level of altruists. For each resource
$e$ we then test all possible combinations of agents from the
different levels that we can allocate to a resource $e$ and satisfy
all bounds. The matrix $R$ changes in dimension to $(n_{\beta_1}+1)
\times \ldots \times (n_{\beta_k}+1)$ to account for all feasible
combinations of remaining agents. Finally, we need to test all
combinations of delay bounds. However, if $k$ is constant, all these
operations can be done in polynomial time.

\begin{cor}
  For symmetric singleton games with altruists and a constant number
  of different altruism levels, there is a polynomial time algorithm
  to decide if a Nash equilibrium exists and to compute the best and
  the worst Nash equilibrium.
\end{cor}

As a byproduct, our approach also allows us to compute a social
optimum state in polynomial time. We simply assume all agents to be
pure altruists and compute the best Nash equilibrium.
\begin{cor}
  For symmetric singleton congestion games a social optimum state can
  be obtained in polynomial time.
\end{cor}
In case of asymmetric games, however, deciding the existence of Nash
equilibria becomes significantly harder.

\begin{thm}
\label{singletonNP}
It is \NP-hard to decide if a singleton congestion game with only pure
altruists and egoists has a Nash equilibrium if G is asymmetric and
has concave delay functions.
\end{thm}

\begin{proof}
  We reduce from {\sc 3Sat}. Given a formula $\varphi$, we construct a
  congestion game $G_{\varphi}$ that has a Nash equilibrium if and
  only if $\varphi$ is satisfiable. Let $x_1,\ldots,x_n$ denote the
  variables and $c_1,\ldots,c_m$ the clauses of a formula $\varphi$.
  Without loss of generality~\cite{Tovey84}, we assume each
  variable appears at most twice positively and at most twice
  negatively.

  For each variable $x_i$ there is a selfish agent $X_i$ that chooses
  one of the resources $e^1_{x_i}$, $e^0_{x_i}$, or $e_0$. The
  resources $e^1_{x_i}$ and $e^0_{x_i}$ have the delay function $9x$
  and resource $e_0$ has the delay function $7x +
  3$. 
  For each clause $c_j$, there is a selfish agent $C_j$ who can choose
  one of the following three resources.  For every positive literal
  $x_i$ in $c_j$ he may choose $e^0_{x_i}$. For every negated literal
  $\bar{x_i}$ in $c_j$ he may choose $e^1_{x_i}$.  Note that there is
  a stable configuration with no variable agent on $e_0$ if and only
  if there is a satisfiable assignment for $\varphi$.  Additionally,
  there are three selfish agents $u_1$, $u_2$, and $u_3$ who can
  choose $e_1$ or $e_2$. Each of the resources $e_1$ and $e_2$ has
  delay $4$ if used by one agent, delay $8$ if used by two agents and
  delay $9$ otherwise. The only pure altruist $u_0$ chooses between
  $e_1$, $e_2$, and $e_0$. Note that the altruist chooses $e_1$, $e_2$
  if one of the variable agents is on $e_0$.

  If $\varphi$ is satisfiable by a bitvector $(x^*_1,\ldots,x^*_n)$, a
  stable solution for $G_\varphi$ can be obtained by placing each
  variable agent $x_i$ on $e^{x^*_i}_{x_i}$. Since
  $(x^*_1,\ldots,x^*_n)$ satisfies $\varphi$ there is one resource for
  each clause agent that is not used by a variable agent. Thus, we can
  place each clause agent on this resource, which he then shares with
  at most one other clause agent.  Let the altruist $u_0$ use $e_0$
  and $u_1$ and $u_2$ choose $e_1$ and $u_3$ choose $e_2$. It is easy
  to check that this is a Nash equilibrium.

  If $\varphi$ is unsatisfiable, there is no stable solution. To prove
  this it suffices to show that one of the variable agents prefers
  $e_0$. In that case the altruist never chooses $e_0$ and the agent
  $u_0,\ldots,u_3$ play the sub game of Example~\ref{exm:noNE}.  For
  the purpose of contradiction assume that $\varphi$ is not
  satisfiable but there is a stable solution in which no variable
  wants to choose $e_0$. This implies that there is no other agent,
  i.e. a clause agent, on a resource that is used by a variable
  agent. However, if all clause agents are on a resource without a
  variable agent we can derive a corresponding bit assignment which,
  by construction, satisfies $\varphi$.

  Therefore, $G_{\varphi}$ has a stable solution if and only if
  $\varphi$ is satisfiable.
\end{proof}

\section{General Games}
\label{sect:general}
For any singleton game $G$ with altruists and convex delay functions a
Nash equilibrium always exists. For more general network structures,
we show that convexity of delay functions is not sufficient. In
particular, this holds even for games with only pure altruists and
egoists in the case in which almost all delay functions are linear of
the form $d_e(x) = a_ex$, except for two edges, which have quadratic
delay functions $d_e(x) = a_ex^2$. For simplicity, we use some edges
with non-convex constant delay $b_e$. We can replace these edges by
sufficiently many parallel edges with delay $b_ex$. This
transformation is of polynomial size and yields an equivalent game
with only convex delays.

\begin{thm}
  It is \NP-hard to decide if a symmetric network congestion game with
  only pure altruists and egoists and quadratic delay functions has a
  Nash equilibrium.
\end{thm}

\begin{proof}
  We first reduce from {\sc 3Sat} to asymmetric congestion games.
  Again, we assume each variable appears at most twice positively and
  at most twice negatively. In a second step, we show that the
  resulting congestion games can be turned into symmetric games while
  preserving all necessary properties.

  \begin{figure}[tbp]
  \begin{minipage}{.6\linewidth}
    \centering \resizebox{\linewidth}{!}{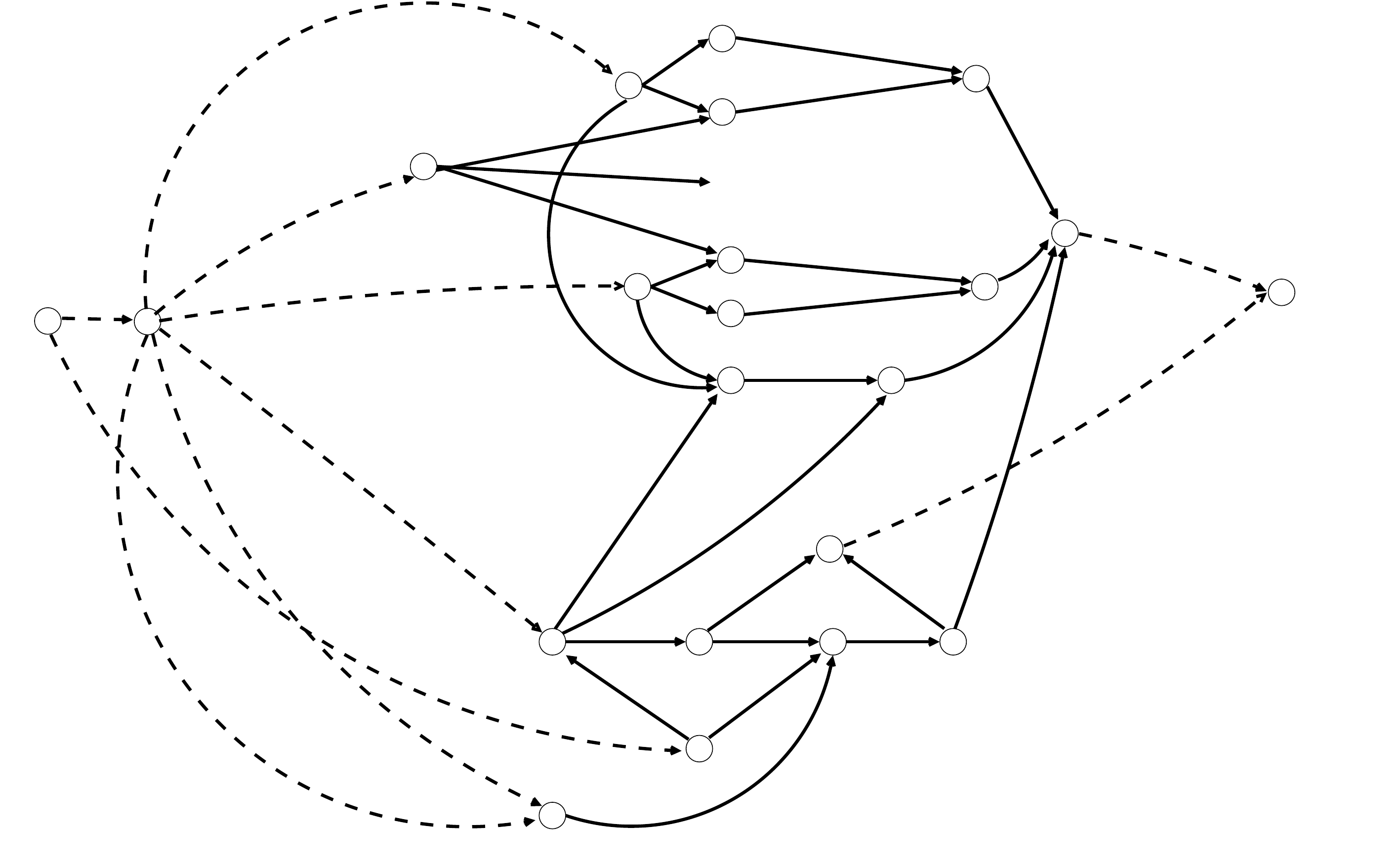}
      \caption{The structure of the 
      network of $G_\Phi$ (solid edges only)
      and $G'_\Phi$.}
    \label{fig:NPnet}
    
   \end{minipage}
   \,
\begin{minipage}{.4\linewidth}  
\begin{tabular}{|l|l|}
\hline
Edge & delay function \\
\hline
$e_0$ & $7x+3$ \\
$e_1$ & $2$ \\
$e_2$ & $17$ \\
$e_4$ & $2.4 x^2$ \\
$e_6$ & $x^2$ \\
$e_{10}$ & $18.5$ \\
$e^1_{x_i}$,$e^1_{x_i}$  & $9x$ \\
$(s,s_{x_i})$  $\forall 1 \le i \le n$ & $ Mx$ \\
$(s,s_{c_j})$  $\forall 1 \le j \le m$ & $ Mx$ \\
$(s,s_1)$, $(s,s_2)$, $(s,s')$ & $ Mx$ \\
$(s,s_0)$ & $ (n+m+5) M$ \\
$(t_0,t)$ & $ (n+m+5) M$ \\
$(t',t)$ & $ Mx$ \\
\hline
\end{tabular}
\captionof{table}{The delay functions on the edges of $G_\Phi$ and
  $G'_\Phi$. Edges that are not listed here have delay of $0$.}
    \end{minipage}

  \end{figure}
  Our reduction is similar to the construction that we used in the
  proof of Theorem~\ref{singletonNP}.  The structure of the resulting
  network congestion game $G_\Phi$ is depicted in
  Figure~\ref{fig:NPnet}.

  Each agent $X_i$ chooses one of three paths from his source node
  $s_{x_i}$ to his target node $t'$ and therefore uses exactly one of
  the edges $e^0_{x_i}$, $e^1_{x_i}$, or $e_0$.  Each clause agent
  $C_j$ uses a path from $s_{c_j}$ to $t'$ and uses one of the three
  edges as described in the proof of Theorem~\ref{singletonNP}. That
  is, for each positive literal $x_i$ in $c_j$ he may choose a path
  that includes the edge $e^0_{x_i}$. For every negated literal
  $\bar{x_i}$ in $c_j$ he may choose a path that contains the edge
  $e^1_{x_i}$.

  There is a selfish agent $u_1$ that chooses a path from $s_{1}$ to
  $t'$ and two selfish agents $u_2$ and $u_3$ that allocate the path
  from $s_{2}$ to $t'$.  Finally, one altruistic agent $u_0$ chooses
  a path from $s_0$ to $t_0$.  As in the proof of
  Theorem~\ref{singletonNP}, we can conclude that there is a variable
  agent whose best response includes edge $e_0$ if and only if $\Phi$
  is not satisfiable.  If no variable agent is on $e_0$, a Nash
  equilibrium can be obtained by placing $u_1$ on the path that begins
  with $(e_1,e_0)$. For agent $u_0$ it is optimal to choose the path
  $(e_8,e_6,e_3)$.

  However, if at least one variable agent is on the edge $e_0$, there
  is no Nash equilibrium.  If the altruist $u_0$ is on the path
  $(e_8,e_6,e_3)$, the best response for $u_1$ is the path
  $(e_4,e_5,e_6,e_7)$.
  If $u_1$ is the path $(e_4,e_5,e_6,e_7)$, the best response for the
  altruist $u_0$ is path $(e_7,e_4,e_2)$. If $u_0$ is on
  $(e_7,e_4,e_2)$, the best response for $u_1$ is the path that begins
  with the edge $e_{10}$. This, finally, is a state in which the best
  response for $u_0$ is $(e_8,e_6,e_3)$.  Thus, the constructed
  network congestion game $G_\Phi$ has a Nash equilibrium if and only
  if the formula $\Phi$ is satisfiable.

  Now, we turn the asymmetric network congestion game $G_\Phi$ into a
  symmetric congestion game $G'_\Phi$. We add a new source node $s$, a
  new target node $t$ and a node $s'$ to the network and connect them
  to $G_\Phi$ as depicted by the dashed edges in
  Figure~\ref{fig:NPnet}. Note that $M$ is an integer that is larger
  than the sum of possible delay values in $G_\Phi$.  If all agents
  play their best responses, then we can observe the following: Each
  outgoing edge of $s'$ is used by exactly one selfish agent and the
  altruist chooses a path that begins with the edge $(s,s_0)$. Every
  best response path of a selfish agent finishes with the edge
  $(t',t)$. Every best response path of the altruist ends with the
  edge $(t_0,t)$. Therefore, $G'_\Phi$ has a Nash equilibrium if and
  only if $G_\Phi$ has a Nash equilibrium.
\end{proof}

Perhaps surprisingly, if \emph{every} delay function is linear $d_e(x)
= a_ex + b_e$, then an elegant combination of the Rosenthal potential
and the social cost function yields a potential for arbitrary
$\beta_i$-altruists. Hence, existence of Nash equilibria and
convergence of sequential better-response dynamics is always
guaranteed. The proof is carefully constructed for altruists, as for
congestion games with general player-specific linear latency functions
a potential does not exist~\cite{Gairing06}. We only consider delays
$d_e(x) = a_ex$ without offset $b_e$, but as noted earlier, this is not
a restriction.

\begin{thm}
  For any congestion game with altruists and linear delay functions
  there is always a Nash equilibrium and sequential better-response
  dynamics converges.
\end{thm}

\begin{proof}
  The theorem follows from the existence of a weighted potential
  $\Phi$ that decreases during every improvement step of any agent
  $i$ with altruism level $\beta_i$.

  \[ \Phi(S) = \sum_{e \in E} \sum_{j=1}^{n_e} a_e j + \sum_{e \in E}
  a_e n_e^2 - \sum_{i=1}^{n} \sum_{e \in S_i}
  \frac{2\beta_i-1}{\beta_i+1} a_e\]

  Consider a state $S$ and an improving strategy change of an agent
  $i$ from $S_i$ to $S'_i$ resulting in a strategy profile $S'$.  We
  show that $\Phi$ decreases. For the sake of clarity and brevity we
  set $\Delta_N = \DeltaPhi$ and $\Delta_C = \DeltaC$. Note that an
  improving strategy change requires $\left(1-\beta\right)\Delta_N +
  \beta \Delta_C > 0$.

  \begin{align*}
    \Phi(S) - \Phi(S') &=
    \Delta_N + \Delta_C - \sum_{e \in S_i \setminus S'_i}  \frac{2\beta_i-1}{\beta_i+1} a_e
    + \sum_{e \in S'_i \setminus S_i} \frac{2\beta_i-1}{\beta_i+1} a_e \\
    &= \left(1-\frac{2(2\beta_i-1)}{1+\beta_i}\right)  \Delta_N  + \Delta_C + \frac{2(2\beta_i-1)}{1+\beta_i} \Delta_N - \sum_{e \in S \setminus S'_i} \frac{2\beta_i-1}{\beta_i+1} a_e
    + \sum_{e \in S'_i \setminus S} \frac{2\beta_i-1}{\beta_i+1} a_e \\
    &= \left(1-\frac{2(2\beta_i-1)}{1+\beta_i}\right) \Delta_N
    + \Delta_C + \frac{(2\beta_i-1)}{1+\beta_i} \Delta_C
  \end{align*}
\[
    = \frac{3-3\beta_i}{1+\beta_i} \Delta_N  + \frac{3\beta_i}{1+\beta_i} \Delta_C =\frac{3}{1+\beta_i} \left( \left(1-\beta_i\right)\Delta_N  + \beta_i \Delta_C \right) > 0
\]
\end{proof}

Unfortunately, it follows directly from previous work~\cite{Fabri04}
that the number of iterations to reach a Nash equilibrium can be
exponential, and the problem of computing a Nash equilibrium is
\PLS-hard. For regular congestion games with matriod strategy
spaces~\cite{Acker06} Nash dynamics converge in polynomial time. It is
an interesting open problem if a similar result holds here.

\section{Stabilization Methods}
\label{sect:stabilize}
This section treats a model in which an institution can convince
selfish agents to act as altruists. For simplicity of presentation we
first restrict to games with only pure altruists and egoists. A
natural question for such an institution to consider is how many
altruists are required to guarantee that there is a Nash equilibrium
with a certain cost, e.g.~a Nash equilibrium as cheap as a social
optimum state. A similar question has been considered for Stackelberg
routing in the Wardrop model~\cite{Kaporis06,Sharma07}. We term this
number the \emph{optimal stability threshold}. In a more pessimistic
direction it is of interest to determine the minimum number of
altruists needed to guarantee that the worst-case Nash equilibrium is
optimal. We term this number the \emph{optimal anarchy threshold}. Let
us denote by $n_1^+$ and $n_1^-$ the optimal stability and anarchy
threshold, respectively. As a consequence from
Theorem~\ref{thm:singleALLES} we can compute both numbers for
symmetric singleton congestion games in polynomial time. For each
number of altruists we check if the best and/or worst Nash equilibrium
is as cheap as the social optimum.

\begin{cor}
  For symmetric singleton congestion games with only pure altruists
  and egoists there is a polynomial time algorithm to compute $n_1^+$
  and $n_1^-$.
\end{cor}

Note that the optimal anarchy threshold is not well-defined, because
the worst Nash equilibrium might always be suboptimal, even for a
population of altruists only. In case of symmetric singleton games and
convex delay functions, an easy exchange argument serves to show that
in this case any local optimum is also a global optimum. However, for
concave delay functions or asymmetric singleton games, a local optimum
might still be globally suboptimal.\footnote{
  Consider a symmetric game with two resources, $d_1(1) = 16$, $d_1(2)
  = 32$, $d_1(3) = 36$, and $d_2(x) = 45$. If all agents allocate
  resource 1, we get a Nash equilibrium of cost 108. In the optimum
  two agents allocate resource 2 resulting in a cost of 106. Now
  consider an asymmetric game with three resources and delay functions
  $d_1(x) = d_2(x) = 8x$, and $d_3(x) = 4x$. Agent 1 can use
  resources 1 and 2, agents 2 and 3 can use resources 2 and 3. The
  state $(2,3,3)$ is a Nash equilibrium of cost 32, while the social
  optimum is a state $(1,2,3)$ of cost 20.} Note that for symmetric
games, our algorithm is able to detect the cases in which suboptimal
local optima exist. In the asymmetric case, however, a similar
approach fails, because of the \NP-hardness of determining existence
of a Nash equilibrium. Thus, in the following we concentrate on the
optimal stability threshold.

In asymmetric games, it is also required to determine the identity of
agents, so here we strive to find a set (denoted $N_e^+$) of minimum
cardinality. For an optimal set of congestion values $n_E^* =
(n_e^*)_{e \in E}$ we can determine $N_1^+(n_E^*)$ such that there is
a Nash equilibrium of the game with congestion values $n_e^*$ for all
$e \in E$.

\begin{thm}
  For singleton games with only pure altruists and egoists and a
  social optimal congestion vector $n_E^*$ there is a polynomial time
  algorithm to compute $N_1^+(n_E^*)$.
\end{thm}

\begin{proof}
  Suppose we are given a congestion vector $(n_e^*)_{e \in E}$ that
  results in minimum social cost. We now construct a weighted
  bipartite graph as follows. One partition is the set of agent
  $N$. In the other partition we introduce for each resource $e$ a
  number of $n_e^*$ vertices. If $e \in \stratSpace_i$ we connect
  agent $i$ to all vertices that were introduced due to $e$. If $e$
  represents a best-response for $i$, then we assign a weight of 0 to
  all corresponding edges between $i$ and the vertices of $e$. To all
  other edges we assign a weight of 1. Note that any feasible
  allocation of agents to strategies that generates the congestion
  vector $n_e^*$ is represented by a perfect matching. Due to social
  optimality an altruist can be matched with any strategy, while an
  egoist must be matched to a best response. If we match an agent to a
  strategy, which is not a best-response, it thus has to become an
  altruist and a weight of one is counted towards the weight of the
  matching. By computing a minimum weight perfect
  matching~\cite{Cook99}, we can identify a minimal set $N_1^+(n_E^*)$
  of altruists required.
\end{proof}

Observe that by creating the edges of cost 1 only to strategies which
represent best-responses with respect to the altruistic delay $d'$, we
can compute $N_1^+(n_E)$ for arbitrary congestion vectors $n_E$. In
this case, the set might be empty, if e.g.~the congestion vector
corresponds to a very expensive state and can never be generated by a
Nash equilibrium for any distribution of altruists. This case,
however, can be recognized by the absence of a perfect matching in the
bipartite graph.

This approach turns out to be applicable to an even more general
natural scenario. Suppose each agent $i$ has a \emph{stability cost}
$c_{ie}$ for each strategy $e \in \stratSpace_i$. This cost yields the
disutility for being forced to play a certain strategy given a
congestion vector $n_E$. In this scenario we slightly change
$N_1^*(n_E)$ to the set agents of minimal stability cost. Still, we
can compute this set by a minimum weight perfect matching if we set
the weights to $c_{ie}$ for all edges connecting $i$ to vertices of
$e$. The stability cost allows for general preferences exceeding
categories like altruists and egoists.

\begin{cor}
  For singleton games and a congestion vector $n_E$ there is a
  polynomial time algorithm to compute $N_1^+(n_E)$ with minimal
  stability cost.
\end{cor}

The underlying problem can be seen as a slot allocation to agents. As
the computed allocation has minimal stability cost, it is possible to
turn the algorithm into a truthful mechanism using VCG payments (see
e.g.~\cite[chapter 9]{Nisan07}). Our final mechanism (1) learns the
stability costs from each agent, (2) determines the allocation, and
(3) pays appropriate amounts to agents for truthful revelation of cost
values and adaptation of allocated strategies. In addition, it can be
verified that all computations needed require only polynomial time.

\begin{cor}
  For singleton games and a congestion vector $n_E$ there is a
  truthful VCG-mechanism to compute $N_1^+(n_E)$ in polynomial time.
\end{cor}

These general results are restricted to the case of singleton
games. For more general games we show that it is \NP-hard to decide if
there is a Nash equilibrium as cheap as the social optimum. Our next
theorem establishes this even for symmetric network congestion games
with linear delays, in which an arbitrary Nash equilibrium and a
social optimum state can be computed in polynomial
time~\cite{Fabri04}. Furthermore, the result requires only a
series-parallel network. Thus, even in this restricted case it is
\NP-hard to decide if the number $n_1^+$ of pure altruists required is
0 or 1, or equivalently if $N_1^+(n_E^*)$ is empty or not. 

\begin{thm}
  For symmetric network congestion games with 3 agents, linear delay
  functions on series-parallel graphs and optimal congestions $n_E^*$
  it is \NP-hard to decide if there is a Nash equilibrium with
  congestions $n_E^*$.
\end{thm}

\begin{proof}
  We reduce from {\sc Partition}. Let an instance be given by positive
  integers $a_1,\ldots,a_k$ and $a = \sum_{i=1}^k a_i$, where $a$ is
  an even number. Create a network with two nodes and two parallel
  edges $e_1$ and $e_2$ for each integer $a_i$. The delay $d_{e_1}(x)
  = 2a_ix$, and $d_{e_2}(x) = a_ix$. All these networks are
  concatenated sequentially. We denote the first node of this path
  gadget by $u$ and the last by $v$. In addition, we add one edge $f =
  (u,v)$ with delay $d_f(x) = \frac 34 ax$. Finally, the game has
  three egoists, which need to allocate a path from $u$ to $v$.

  The unique social optimum is to let one agent use $f$ and the other
  two agents use two edge-disjoint paths through the path
  gadget. This yields an optimal social cost of
  $\frac{15}{4}a$. However, for a Nash equilibrium each path through
  the gadget must not have more delay than $\frac 32 a$. If the
  instance of {\sc Partition} is solvable, then the elements assigned
  to a partition represent the edges of type $e_1$ that an agent
  allocates in Nash equilibrium. Otherwise, if the instance is not
  solvable, there is no possibility to partition the path gadget into
  two edge-disjoint paths of latency at most $\frac 32 a$.
  
  The reduction works for a small constant number of agents but only
  shows weak \NP-hardness. If the number of agents is variable, it is
  possible to show strong \NP-hardness with a similar reduction from
  {\sc 3-Partition}.
\end{proof}

We remark that the previous theorem contrasts the continuous
non-atomic case, in which a minimal fraction of altruistic demand
stabilizing an optimum solution can be computed in any symmetric
network congestion game~\cite{Kaporis06}.

\section{Conclusions}
In this paper, we have initiated the study of altruists in atomic
congestion games. Our model is similar to the one presented by Chen
and Kempe~\cite{Chen08} for nonatomic routing games, however, we
observe quite different properties. In the nonatomic case, existence
of Nash equilibria for any population of agents is always guaranteed,
even if agents are partially spiteful. In contrast, our study answers
fundamental questions for existence and convergence in atomic
games. For the case of linear latencies, an elegant combination of
social cost and the Rosenthal potential proves guaranteed existence
and convergence. The next step is to consider the price of anarchy and
the relations to results on Stackelberg games~\cite{Fotakis07}. An
altruistic variant of the price of
malice~\cite{Moscibroda06Evil,Meier08} measuring the influence of
altruists on the worst-case Nash equilibrium can be interesting to
consider. Finally, a characterization of games for which Nash
equilibria exist and best-response dynamics converge (in polynomial
time) is an important open problem.

\bibliographystyle{plain}
\bibliography{../../Bibfiles/game,../../Bibfiles/price,../../Bibfiles/mechanism,../../Bibfiles/steiner}
\end{document}

%% file: network2.pdf_t
\begin{picture}(0,0)%
\includegraphics{network2.pdf}%
\end{picture}%
\setlength{\unitlength}{4144sp}%
\begingroup\makeatletter\ifx\SetFigFontNFSS\undefined%
\gdef\SetFigFontNFSS#1#2#3#4#5{%
  \reset@font\fontsize{#1}{#2pt}%
  \fontfamily{#3}\fontseries{#4}\fontshape{#5}%
  \selectfont}%
\fi\endgroup%
\begin{picture}(12954,8029)(-1028,-7416)
\put(10931,-2223){\makebox(0,0)[lb]{\smash{{\SetFigFontNFSS{25}{30.0}{\rmdefault}{\mddefault}{\updefault}{\color[rgb]{0,0,0}$t$}%
}}}}
\put(8894,-1412){\makebox(0,0)[lb]{\smash{{\SetFigFontNFSS{25}{30.0}{\rmdefault}{\mddefault}{\updefault}{\color[rgb]{0,0,0}$t'$}%
}}}}
\put(4471,-2451){\makebox(0,0)[lb]{\smash{{\SetFigFontNFSS{25}{30.0}{\rmdefault}{\mddefault}{\updefault}{\color[rgb]{0,0,0}$s_{x_1}$}%
}}}}
\put(6276,-3218){\makebox(0,0)[lb]{\smash{{\SetFigFontNFSS{25}{30.0}{\rmdefault}{\mddefault}{\updefault}{\color[rgb]{0,0,0}$e_0$}%
}}}}
\put(5569,-6374){\makebox(0,0)[lb]{\smash{{\SetFigFontNFSS{25}{30.0}{\rmdefault}{\mddefault}{\updefault}{\color[rgb]{0,0,0}$s_0$}%
}}}}
\put(3749,-7280){\makebox(0,0)[lb]{\smash{{\SetFigFontNFSS{25}{30.0}{\rmdefault}{\mddefault}{\updefault}{\color[rgb]{0,0,0}$s_{2}$}%
}}}}
\put(4471,-6061){\makebox(0,0)[lb]{\smash{{\SetFigFontNFSS{25}{30.0}{\rmdefault}{\mddefault}{\updefault}{\color[rgb]{0,0,0}$e_7$}%
}}}}
\put(5825,-5565){\makebox(0,0)[lb]{\smash{{\SetFigFontNFSS{25}{30.0}{\rmdefault}{\mddefault}{\updefault}{\color[rgb]{0,0,0}$e_5$}%
}}}}
\put(6908,-5565){\makebox(0,0)[lb]{\smash{{\SetFigFontNFSS{25}{30.0}{\rmdefault}{\mddefault}{\updefault}{\color[rgb]{0,0,0}$e_6$}%
}}}}
\put(6141,-4933){\makebox(0,0)[lb]{\smash{{\SetFigFontNFSS{25}{30.0}{\rmdefault}{\mddefault}{\updefault}{\color[rgb]{0,0,0}$e_2$}%
}}}}
\put(6817,-4933){\makebox(0,0)[lb]{\smash{{\SetFigFontNFSS{25}{30.0}{\rmdefault}{\mddefault}{\updefault}{\color[rgb]{0,0,0}$e_3$}%
}}}}
\put(4696,-4572){\makebox(0,0)[lb]{\smash{{\SetFigFontNFSS{25}{30.0}{\rmdefault}{\mddefault}{\updefault}{\color[rgb]{0,0,0}$e_1$}%
}}}}
\put(3703,-5610){\makebox(0,0)[lb]{\smash{{\SetFigFontNFSS{25}{30.0}{\rmdefault}{\mddefault}{\updefault}{\color[rgb]{0,0,0}$s_{1}$}%
}}}}
\put(4696,-5565){\makebox(0,0)[lb]{\smash{{\SetFigFontNFSS{25}{30.0}{\rmdefault}{\mddefault}{\updefault}{\color[rgb]{0,0,0}$e_4$}%
}}}}
\put(6502,-3895){\makebox(0,0)[lb]{\smash{{\SetFigFontNFSS{25}{30.0}{\rmdefault}{\mddefault}{\updefault}{\color[rgb]{0,0,0}$e_{10}$}%
}}}}
\put(6095,-5971){\makebox(0,0)[lb]{\smash{{\SetFigFontNFSS{25}{30.0}{\rmdefault}{\mddefault}{\updefault}{\color[rgb]{0,0,0}$e_8$}%
}}}}
\put(6647,-650){\makebox(0,0)[lb]{\smash{{\SetFigFontNFSS{25}{30.0}{\rmdefault}{\mddefault}{\updefault}{\color[rgb]{0,0,0}$e^1_{x_{n}}$}%
}}}}
\put(6647,253){\makebox(0,0)[lb]{\smash{{\SetFigFontNFSS{25}{30.0}{\rmdefault}{\mddefault}{\updefault}{\color[rgb]{0,0,0}$e^0_{x_{n}}$}%
}}}}
\put(4571, 72){\makebox(0,0)[lb]{\smash{{\SetFigFontNFSS{25}{30.0}{\rmdefault}{\mddefault}{\updefault}{\color[rgb]{0,0,0}$s_{x_{n}}$}%
}}}}
\put(5657,-7007){\makebox(0,0)[lb]{\smash{{\SetFigFontNFSS{25}{30.0}{\rmdefault}{\mddefault}{\updefault}{\color[rgb]{0,0,0}$e_9$}%
}}}}
\put(6727,-1774){\makebox(0,0)[lb]{\smash{{\SetFigFontNFSS{25}{30.0}{\rmdefault}{\mddefault}{\updefault}{\color[rgb]{0,0,0}$e^0_{x_1}$}%
}}}}
\put(6727,-2496){\makebox(0,0)[lb]{\smash{{\SetFigFontNFSS{25}{30.0}{\rmdefault}{\mddefault}{\updefault}{\color[rgb]{0,0,0}$e^1_{x_1}$}%
}}}}
\put(6930,-1308){\rotatebox{90.0}{\makebox(0,0)[lb]{\smash{{\SetFigFontNFSS{45}{54.0}{\rmdefault}{\mddefault}{\updefault}{\color[rgb]{0,0,0}...}%
}}}}}
\put(6141,-4497){\makebox(0,0)[lb]{\smash{{\SetFigFontNFSS{25}{30.0}{\rmdefault}{\mddefault}{\updefault}{\color[rgb]{0,0,0}$t_0$}%
}}}}
\put(-78,-2645){\makebox(0,0)[lb]{\smash{{\SetFigFontNFSS{25}{30.0}{\rmdefault}{\mddefault}{\updefault}{\color[rgb]{0,0,0}$s'$}%
}}}}
\put(2957,-692){\rotatebox{90.0}{\makebox(0,0)[lb]{\smash{{\SetFigFontNFSS{45}{54.0}{\rmdefault}{\mddefault}{\updefault}{\color[rgb]{0,0,0}...}%
}}}}}
\put(2814,-1445){\makebox(0,0)[lb]{\smash{{\SetFigFontNFSS{25}{30.0}{\rmdefault}{\mddefault}{\updefault}{\color[rgb]{0,0,0}$s_{c_j}$}%
}}}}
\put(8318,-4291){\makebox(0,0)[lb]{\smash{{\SetFigFontNFSS{25}{30.0}{\rmdefault}{\mddefault}{\updefault}{\color[rgb]{0,0,0}$e_7$}%
}}}}
\put(-1013,-2489){\makebox(0,0)[lb]{\smash{{\SetFigFontNFSS{25}{30.0}{\rmdefault}{\mddefault}{\updefault}{\color[rgb]{0,0,0}$s$}%
}}}}
\end{picture}%